\documentclass[a4paper,12pt]{amsart}

\usepackage{amsmath}
\usepackage{amsfonts}
\usepackage{amssymb}
\usepackage{amsthm}

\pdfoutput=1

\DeclareMathOperator{\Pert}{Pert}

\def\A{\mathcal{A}}
\def\B{\mathcal{B}}
\def\C{\mathbb{C}}
\def\H{\mathbb{H}}
\def\R{\mathbb{R}}
\def\U{\mathcal{U}}

\title[Perturbation semigroup of matrix algebras]{Perturbation semigroup \\of matrix algebras}
\author{Niels Neumann and Walter D. van Suijlekom}
\address{Institute for Mathematics, Astrophysics and Particle Physics, 
Faculty of Science, Radboud University Nijmegen, Heyendaalseweg 135, 6525 AJ Nijmegen, The Netherlands}
\email{niels.neumann@student.ru.nl, waltervs@math.ru.nl}
\date{22 October 2014}

\theoremstyle{plain}		\newtheorem{theorem}{Theorem}[section]
\theoremstyle{definition}	\newtheorem{definition}[theorem]{Definition}
\theoremstyle{plain} 		\newtheorem{lemma}[theorem]{Lemma}

\newtheorem{proposition}[theorem]{Proposition}
\newtheorem{corl}[theorem]{Corollary}

\newtheorem{corollary}[theorem]{Corollary}					
\newtheorem{remark}[theorem]{Remark}

\begin{document}

\begin{abstract}
In this article we analyze the structure of the semigroup of inner perturbations in noncommutative geometry. This perturbation semigroup is associated to a unital associative $*$-algebra and extends the group of unitary elements of this $*$-algebra. We compute the perturbation semigroup for all matrix algebras. 
\end{abstract}

\maketitle

\section{Introduction}
Recently, a semigroup structure has been introduced \cite{CCS13} in the context of noncommutative geometry \cite{C94}. This {\em perturbation semigroup} is associated to a (unital associative) $*$-algebra $\A$, and implements the inner perturbations \cite{C96} of the metric ---described in terms of a `Dirac operator' $D$ acting on a Hilbert space $\mathcal H$--- in a spectral triple $(\A,\mathcal H,D)$ ({\em cf.} \cite{C95}). Moreover, the perturbation semigroup allows for a description of such fluctuations when the so-called first-order condition is not satisfied. The physical applications requiring such an extension were subsequently discussed in \cite{CCS13b}. A crucial role in these applications is played by finite spectral triples, that is, spectral triples for which $\A$ and $H$ are finite-dimensional (and accordingly, $D$ is a hermitian matrix). It is the subject of this paper to determine the perturbation semigroup for all such finite-dimensional $*$-algebras $\A$. Since $\A$ is faithfully represented on $\mathcal H$, this amounts to considering only matrix algebras. In other words, we consider the $*$-algebra of block-diagonal matrices of the form
\begin{equation}
\label{eq:matrix-alg}
\A = \bigoplus_{i=1}^N M_{n_i} (\mathbb F_i) ,
\end{equation}
where $n_1,\ldots, n_N$ are the fixed dimensions of the block-matrices and $\mathbb F_i= \C,\R$ or $\H$ (which may vary with $i$). We stress that the algebra $\A$ is a complex $*$-algebra only if all $\mathbb F_i = \C$, otherwise we consider it as a real $*$-algebra.   

In Section \ref{sect:pert} we introduce and analyze the general structure of the perturbation semigroup $\Pert(\A)$ associated to a $*$-algebra and show how it extends the group $\U(\A)$ of unitary elements in $\A$. We then show how $\Pert(\A \oplus \mathcal B)$ is related to $\Pert(\A)$ and $\Pert(\mathcal B)$. This allows for a determination of the perturbation semigroup of all matrix algebras by the computation of $\Pert(M_N(\mathbb F))$ for $\mathbb F = \C,\R$ or $\H$ in Section \ref{sect:pert-matrix}. In all these examples we identify the map from the group of unitaries in $\A$ to $\Pert(\A)$, and relate it to the representation theory of $U(N)= \U(M_N(\C))$, $O(N)=\U(M_N(\R))$ and $Sp(N)=\U(M_N(\H))$ in the respective cases. 

\subsection*{Acknowledgement}
The first author is supported by the Radboud Honours Academy of the Faculty of Science. The second author thanks the Hausdorff Institute for Mathematics in Bonn for hospitality and support.

\section{The perturbation semigroup}
\label{sect:pert}
Throughout this paper, we let $\A$ be an associative unital $*$-algebra, referring to it simply as a $*$-algebra. We allow both complex and real $*$-algebras, {\it i.e.} the base field is either $\C$ or $\R$. This is important when we consider tensor products: $\A \otimes \B$ will then denote either $\A \otimes_\C \B$ or $\A \otimes_\R \B$, depending on whether $\A,\B$ are considered as complex or real algebras. 

Associated to any $*$-algebra, we can define a group as follows. 
\begin{definition}
The {\em group of unitary elements} in a $*$-algebra $\A$ will be denoted by $\mathcal{U}(\mathcal{A})$, {\it i.e.}
$$\mathcal{U}(\mathcal{A}) = \lbrace u \in \mathcal{A} \mid uu^* = 1 = u^*u\rbrace.$$
\end{definition}

The unitary group $\U(\A)$ plays the role of a gauge group in noncommutative geometry and its applications to particle physics \cite{C96}. In fact, if $\A$ is represented on a Hilbert space $\mathcal H$, then the unitary group is represented on $\mathcal H$ by, indeed, unitary operators. Moreover, any self-adjoint operator $D$ on $\mathcal H$ can be transformed into a unitarily equivalent operator, via
$$
D \mapsto u D u^*,
$$
which can be rewritten as
$$
D \mapsto u D u^* = D + u [D, u^*].
$$
We interpret this as a perturbation of $D$ by the unitary element $u \in \U(\A)$. A more general class of perturbations associated to $\A$ is given by the perturbation semigroup \cite{CCS13} that we now define. First, we recall the definition of the opposite algebra. 

\begin{definition}
Let $\mathcal{A}$ be an algebra, then the {\em opposite algebra} of $\mathcal{A}$ is denoted by $\mathcal{A}^{\circ}$ and is given by $\mathcal{A}$ as a vector space but with opposite product $a^\circ b^\circ = (ba)^\circ$ for $a,b\in \A$. 
\end{definition}
\begin{definition}
The {\em perturbation semigroup} is given by
$$\Pert(\mathcal{A}) = \left\lbrace \sum_j a_j \otimes b^{\circ}_j \in \mathcal{A} \otimes \mathcal{A}^{\circ} \left|\begin{array}{ll} \sum a_jb_j = 1\\ \sum a_j \otimes b^{\circ}_j = \sum b^*_j \otimes a^{*\circ}_j \end{array}\right. \right\rbrace,$$
where the sums are finite and $1$ is the unit in $\mathcal{A}$. 
\end{definition}
We will refer to the two conditions on the sums in $\Pert(\A)$ as {\em normalization condition} and {\em self-adjointness condition}, respectively.

\begin{proposition}
$\Pert(\mathcal{A})$ is a semigroup and has a unit. 
\end{proposition}
\begin{proof}
Multiplication in $\Pert(\mathcal{A})$ is associative because $\A \otimes \A^\circ$ is an associative algebra. We show that the multiplication is closed, {\em i.e.}~that the product of two elements is again in the perturbation semigroup. For $\sum_j a_j \otimes \widetilde{a}_j^{\circ}, \sum_k b_k \otimes \widetilde{b}_k^{\circ} \in \Pert(\mathcal{A})$, we have
$$
\big(\sum_j a_j \otimes \widetilde{a}_j^{\circ}\big)\big(\sum_k b_k\otimes\widetilde{b}_k^{\circ}\big) 
 = \sum_{j,k} a_jb_k \otimes (\widetilde{b}_k\widetilde{a}_j)^{\circ}.
$$
That this element is both normalized and self-adjoint follows from a simple computation. It is normalized:
$$
\sum_{j,k} (a_jb_k)(\widetilde{b}_k\widetilde{a}_j)  =  \sum_{j,k} a_j (b_k\widetilde{b}_k) \widetilde{a}_j
  =\sum_j a_j (\sum_k b_k \widetilde{b}_k ) \widetilde{a}_j
 =  \sum_j a_j \widetilde{a}_j
 =  1,
$$
because both $\sum_k b_k \otimes \widetilde{b}_k^{\circ}$ and $\sum_j a_j \otimes \widetilde{a}_j^{\circ}$ are normalized, and it is self-adjoint:
\begin{multline*}
\sum_{j,k} (\widetilde{b}_k\widetilde{a}_j)^* \otimes (a_jb_k)^{*\circ}  =  \sum_{j,k} \widetilde{a}_j^* \widetilde{b}_k^* \otimes a_j^{*\circ}b_k^{*\circ}
 =  \big(\sum_j \widetilde{a}_j^* \otimes a_j^{*\circ} \big)\big(\sum_k \widetilde{b}_k^* \otimes b_k^{*\circ} \big)\\
 =  \big(\sum_j a_j \otimes \widetilde{a}_j^{\circ} \big)\big(\sum_k b_k \otimes \widetilde{b}_k^{\circ} \big)
 =  \sum_{j,k} (a_jb_k) \otimes (\widetilde{b}_k\widetilde{a}_j)^{\circ}
\end{multline*}
because both $\sum_k b_k \otimes \widetilde{b}_k^{\circ}$ and $\sum_j a_j \otimes \widetilde{a}_j^{\circ}$ are self-adjoint. The unit is given by $1 \otimes 1^\circ$, where $1$ is the unit in $\mathcal{A}$.
\end{proof}
The name perturbation semigroup is motivated by the following action of $\Pert(\A)$ on self-adjoint operators on $H$:
$$
D \mapsto \sum_j a_j D b_j =  D + \sum_j a_j [D,b_j],
$$
where $\sum_j a_j \otimes b_j^\circ \in \Pert(\A)$. This generalizes the action of $\U(\A)$ on $D$ that we just discussed. In fact, we have the following result.
\begin{proposition}
\label{prop:UA-Pert}
Let $\mathcal{A}$ be a $*$-algebra, then we have
\begin{align}
\label{eq:unit}
\mathcal{U}(\mathcal{A}) & \rightarrow \Pert(\mathcal{A}) \\
u & \mapsto u \otimes u^{*\circ}. \nonumber
\end{align}
\end{proposition}

We end this section by determining the perturbation semigroup of the direct sum of $*$-algebras. 
\begin{proposition}
\label{prop:direct-sum}
Let $\mathcal{A},\mathcal{B}$ be $*$-algebras, then
\begin{equation}
\label{eq:directesom}
\Pert(\mathcal{A}\oplus\mathcal{B}) \cong \Pert(\mathcal{A}) \times \Pert(\mathcal{B}) \times (\mathcal{A}\otimes \mathcal{B}^{\circ}\oplus \mathcal{B}\otimes \mathcal{A}^{\circ})^{{\rm s.a.}}
\end{equation}
where ${\rm s.a.}$ stands for self-adjoint elements, {\it i.e.} those of the form $\sum a_i \otimes b_i^{\circ} + b_i^* \otimes a_i^{*\circ}$.
\end{proposition}
\begin{proof}
We start with the following isomorphism of $*$-algebras:
$$
(\mathcal{A}\oplus \mathcal{B})\oplus (\mathcal{A}\oplus \mathcal{B})^{\circ} \cong   \mathcal{A}\otimes \mathcal{A}^{\circ} \oplus \mathcal{B}\otimes \mathcal{B}^{\circ} \oplus \mathcal{A} \otimes \mathcal{B}^{\circ} \oplus \mathcal{B} \otimes \mathcal{A}^{\circ}.
$$
Imposing the normalization and self-adjointness condition to obtain $\Pert(\A \oplus \mathcal B)$ on the left-hand side translates on the right-hand side to give $\Pert(\mathcal{A}) \times \Pert(\mathcal{B}) \times (\mathcal{A}\otimes \mathcal{B}^{\circ}\oplus \mathcal{B}\otimes \mathcal{A}^{\circ})^{{\rm s.a.}}$. Indeed, normalization only affects the first two terms $\mathcal{A}\otimes \mathcal{A}^{\circ} \oplus \mathcal{B}\otimes \mathcal{B}^{\circ}$ where, together with the self-adjointness condition it gives rise to $\Pert(\A) \times \Pert(\mathcal B)$. The self-adjointness condition on $\mathcal{A} \otimes \mathcal{B}^{\circ} \oplus \mathcal{B} \otimes \mathcal{A}^{\circ}$ gives rise to elements of the form stated above.
\end{proof}

\section{Perturbation semigroup for matrix algebras}
\label{sect:pert-matrix}
In this section, we will derive the structure of the perturbation semigroup for all matrix algebras of the form \eqref{eq:matrix-alg}. In view of Proposition \ref{prop:direct-sum} it is enough to compute $\Pert(M_N(\mathbb F))$ for $\mathbb F = \C,\R,\H$. However, let us start with the following basic example.

\subsection{Perturbation semigroup $\Pert(\mathbb{C}^N)$}
For $\A=\C^N$ we have $\mathcal{A}\cong \mathcal{A}^{\circ}$ 
and $\A \otimes \A^\circ \cong \mathbb{C}^{N^2}$. As a basis for $\mathcal{A}$ we take the standard basis $\lbrace e_i \rbrace_{i=1}^N$. 
\begin{proposition}
\label{prop:pert-CN}
For any $N \geq 1$ we have
$$\Pert(\mathbb{C}^N) \cong \mathbb{C}^{N(N-1)/2}$$
with the semigroup structure given by componentwise multiplication. 
\end{proposition}
\begin{proof}
In terms of the above standard basis we write an arbitrary element in $\Pert(\C^N)$ as $\sum C_{ij} e_i \otimes e_j^{\circ}$. The normalization condition states that $C_{ii}=1$ for all $i$: indeed 
$$\sum_{i,j} C_{ij}e_ie_j = \sum_i C_{ii} e_i,$$
which should be equal to $\sum_i e_i =1$, the unit in $\mathbb{C}^N$. 
The self-adjointness condition states that $C_{ij} = \overline{C_{ji}}$ for all $i,j$, as follows from
\begin{align*}
\sum_{i,j} C_{ij} e_i\otimes e_j^{\circ}  = \sum_{i,j} C_{ij}^* e_j^*\otimes e_i^{*\circ}
= \sum_{i,j} \overline{C_{ij}} e_j \otimes e_i 
= \sum_{i,j} \overline{C_{ji}} e_i \otimes e_j^{\circ}.
\end{align*}
In other words, among the $N^2$ variables $C_{ij}$, $N$ are equal to one, while the others are pairwise conjugated. 
\end{proof}
Note that this is compatible with Proposition \ref{prop:direct-sum}. In fact, with $\C^N \cong \C^{N-1} \oplus \C$ and the fact that $\Pert(\C) = \{ 1\}$ we find that
$$
\Pert(\C^N) \cong \Pert(\C^{N-1}) \times \C^{N-1}
$$
thus giving a different proof of Proposition \ref{prop:pert-CN}.


We end this subsection by considering the map from the unitary group to the perturbation semigroup of $\C^N$, leaving its elementary proof to the reader. 
\begin{proposition}
The map $\mathcal{U}(\mathbb{C}^N) \rightarrow \Pert(\mathbb{C}^N)$ is given explicitly by 
$$
(\lambda_1, \hdots, \lambda_N) \mapsto 1 + \sum_{i  \neq j} \lambda_i \overline{\lambda_j} e_i \otimes e_j^\circ.
$$

\end{proposition}

\subsection{Perturbation semigroup of complex matrix algebras}
We determine the perturbation semigroup of $M_N(\mathbb{C})$. 

\begin{lemma}
\label{lma:matrix-opposite}
We have the following identification
\begin{align*}
M_N(\mathbb{C})^{\circ} & \rightarrow M_N(\mathbb{C})\\
A^{\circ} & \mapsto A^{\sf T}.
\end{align*}
where ${\sf T}$ denotes matrix transposition. Consequently, 
$$
M_N(\C) \otimes M_N(\C)^\circ \cong M_{N^2}(\C),
$$
as $*$-algebras.
\end{lemma}

Under this identification we thus have
$$e_{ij}^{\circ} \leftrightarrow e_{ji}$$
in terms of the standard basis $\lbrace e_{ij}\rbrace_{i,j=1}^N$ for $M_N(\C)$. 
We will write an arbitrary element in $M_N(\mathbb{C})\otimes M_N(\mathbb{C})^{\circ}$ as a sum
$$
\sum_{i,j,k,l} C_{ij,kl} e_{ij} \otimes e_{kl}^{\circ}.
$$

\subsubsection{$\mathcal{A}=M_2(\mathbb{C})$}
\label{subsec:M_2C}
As a warming-up, we first look at $\mathcal{A}=M_2(\mathbb{C})$. Note that we have four basis elements for which the normalization condition becomes
\begin{align*}
&(C_{11,11} + C_{12,21})e_{11}  + (C_{11,12} + C_{12,22})e_{12} \\&\quad + (C_{21,11} + C_{22,21})e_{21}
 + (C_{21,12} + C_{22,22})e_{22} = e_{11} + e_{22} .
\end{align*}
This amounts to the conditions
\begin{align*}
C_{11,11} + C_{12,21} = 1, \quad C_{21,12} + C_{22,22} = 1, \\
C_{11,12} + C_{12,22} = 0, \quad C_{21,11} + C_{22,21} = 0. 
\end{align*}
The self-adjointness condition reads $C_{ij,kl} = \overline{C_{lk,ji}}$ ({\it cf.} Section \ref{sec:genmat} below).

Using Lemma \ref{lma:matrix-opposite} we can identify
\begin{align*}
M_2(\mathbb{C}) \otimes M_2(\mathbb{C})^{\circ} & \rightarrow M_4(\mathbb{C}), \qquad e_{ij} \otimes e_{kl}^{\circ} \mapsto e_{ij} \otimes e_{lk}.
\end{align*}
in terms of the basis elements, and then extend this linearly to all of $M_2(\mathbb{C})\otimes M_2(\mathbb{C})^{\circ}$. The normalization and self-adjointness conditions on $C_{ij,kl}$ translate to $4 \times 4$-matrices to arrive at the following general form for an element $A \in \Pert(M_2(\mathbb{C}))$:
\begin{equation}
\label{eq:gen-form-A}
A =\begin{pmatrix}
x_1				& z_3	& \overline{z_3}	& 1-x_1				\\
z_1				& z_2	& \overline{z_5}	& -z_1				\\
\overline{z_1}	& z_5	& \overline{z_2}	& -\overline{z_1}	\\
x_2				& z_4	& \overline{z_4}	& 1-x_2				\\
\end{pmatrix},
\qquad z_1, \hdots z_5 \in \mathbb{C}, \; x_1,x_2 \in \mathbb{R}.\end{equation}
The semigroup law ensures that the product of two such matrices again has this general form, something which is not immediately clear. Let us make this point more transparent and establish conditions on $4 \times 4$ matrices that give rise to the above form. 

For an element $A \in M_4(\mathbb{C})$ to be of the form \eqref{eq:gen-form-A} is equivalent to demanding that
$$A (e_1 + e_4) = (e_1 + e_4),$$
$$\widehat{\Omega} \overline{A} = A \widehat{\Omega}, \qquad \text{where}\; \widehat{\Omega} = 
\begin{pmatrix}
1 & 0 & 0 & 0 \\
0 & 0 & 1 & 0 \\
0 & 1 & 0 & 0 \\
0 & 0 & 0 & 1 \\
\end{pmatrix}.$$
Equivalently, the matrix $\widehat{\Omega}$ can be rewritten as a block matrix
$$\widehat{\Omega} = 
\begin{pmatrix}
e_{11} & e_{21} \\
e_{12} & e_{22} \\
\end{pmatrix} =
\begin{pmatrix}
e_{11}^{\sf T}	& e_{12}^{\sf T}	\\
e_{21}^{\sf T}	& e_{22}^{\sf T}	\\
\end{pmatrix} = 
\sum e_{ij} \otimes e_{ji}.$$
Especially the last identity is useful, since we see that the eigenvectors of $\widehat{\Omega}$ are given by $e_1 \otimes e_1 \pm e_2 \otimes e_2$, with eigenvalue $1$, and $e_1\otimes e_2 \pm e_2 \otimes e_1$, with eigenvalue $1$ and $-1$ depending on the $+$ or $-$ sign. Hence, upon changing to the basis $$\{ e_1 \otimes e_1 + e_2 \otimes e_2, e_1 \otimes e_1 - e_2 \otimes e_2,e_1\otimes e_2 + e_2 \otimes e_1, e_1\otimes e_2 - e_2 \otimes e_1\}$$ of eigenvectors we will get
\begin{equation}
\label{eq:omega-2}
\Omega = \begin{pmatrix}
I_3 & 0 \\
0 & -1 \\
\end{pmatrix}.
\end{equation}
Moreover, the vector $e_1+e_4$ which is left invariant by $A$ is given by $e_1 \otimes e_1 + e_2 \otimes e_2 \in \C^2 \otimes \C^2$, which is also an eigenvector of $\widehat \Omega$. Hence with respect to this new basis we arrive at the following characterization of $\Pert(M_2(\C))$:
$$
\Pert(M_2(\mathbb{C})) \cong \Big\lbrace A \in M_4(\mathbb{C}) \mid A e_1 = e_1 , \;
\Omega \overline{A} = A \Omega\Big \rbrace,$$
with 
$$
e_1 = \begin{pmatrix} 1\\0 \\ 0 \\ 0 \end{pmatrix}, \qquad 
\Omega = \begin{pmatrix}
I_3 & 0 \\
0 & -1 \\
\end{pmatrix}.
$$
It now readily follows that if $A,B\in \Pert(M_2(\C))$ then so is their product $AB$.

More explicitly, elements in $\Pert(M_2(\C))$ are thus $4 \times 4$ matrices of the form
\begin{equation}
\label{eq:block-Pert-M2}
A = \begin{pmatrix}
1 & v_1 & v_2 & i w \\
0 & x_1 & x_2 & iy_1 \\
0 & x_3 & x_4 & iy_2 \\
0 & iy_3 & iy_4 & x_5 \\
\end{pmatrix},
\end{equation}
where $v_1,v_2,w, x_1,\ldots, x_5,y_1,\ldots,y_4 \in \R$.

The invertible elements in the perturbation semigroup are given by the invertible matrices in $M_4(\mathbb{C})$ which moreover fulfill the above conditions. Thus the group of invertible elements is given by
$$
\Pert(M_2(\mathbb{C}))^{\times} \cong \Big\lbrace A \in GL_4(\mathbb{C}) \mid A e_1 = e_1 , \; \Omega \overline{A} = A \Omega\Big \rbrace.
$$

We end this section by showing how the unitary group $U(2) = \U(M_2(\C))$ maps to this group of invertible elements. Recall that there is a group homomorphism
\begin{equation}
\label{eq:group-homo-U2}
\begin{aligned}
U(2) &\to \Pert(M_2(\C))^\times\\
u& \mapsto u \otimes u^{*\circ}.
\end{aligned}
\end{equation}
After identifying $M_2(\C)^\circ$ with $M_2(\C)$ using Lemma \ref{lma:matrix-opposite}, the element $u \otimes u^{*\circ}$ on the right-hand side of \eqref{eq:group-homo-U2} corresponds to the element $u \otimes \overline{u} \in M_2(\C) \otimes M_2(\C)$. In terms of representation theory, this means that the map in \eqref{eq:group-homo-U2} corresponds to the representation of $u \in U(2)$ on the tensor product $\C^2 \otimes \overline{\C^2}$ of the defining representation $\C^2$ and the conjugate representation $\overline{\C^2}$ of $U(2)$. It is well-known that this representation has the following decomposition in irreducible representations:
$$
\C^2 \otimes \overline{\C^2} \cong \C \oplus \C^3
$$
where $\C$ is the trivial representation space of $U(2)$ and $\C^3$ is the complexified adjoint representation space $\mathfrak{su}(2)_\C$. Moreover, the basis vector spanning the trivial representation is given by $e_1 \otimes e_1+ e_2 \otimes e_2$. If we compare this to the basis of eigenvectors for $\widehat{\Omega}$ that we found above, we conclude that the decomposition of $\C^2 \otimes \overline{\C^2}$ into irreducible representations corresponds precisely to the block decomposition of the matrix $A$ in \eqref{eq:block-Pert-M2}.

\subsubsection{$\mathcal{A}=M_N(\mathbb{C})$}
\label{sec:genmat}
With this example in mind we now proceed and determine $\Pert(M_N(\mathbb{C}))$. First note that with Lemma \ref{lma:matrix-opposite} the matrices in the perturbation semigroup $\Pert(M_N(\mathbb{C}))$ will be elements of $M_{N^2}(\mathbb{C})$. Again, we aim for defining conditions on such matrices using a suitable matrix $\widehat{\Omega}$ that are equivalent to the normalization and self-adjointness condition. 
\begin{lemma}
\label{lma:norm.genmat}
Let $A = C_{ij,kl} e_{ij}\otimes e_{kl}^{\circ}$, then the normalization condition is equivalent to
$$\sum_j C_{ij,jl} = \delta^i_l.$$
\end{lemma}
\begin{proof}

For such an element $A$ the normalization condition reads 
$$\sum_{i,j,k,l} C_{ij,kl}e_{ij}e_{kl} = 1 \equiv \sum_i e_{ii},$$
or, equivalently, 
$$\sum_{i,j,k,l} C_{ij,kl} e_{il} \delta^j_k = \sum_{i,j,l} C_{ij,jl} e_{il} = \sum_i e_{ii}.
$$
Reading off the coefficients gives the desired result. 
\end{proof}
\begin{remark}
Note that this result, and hence the normalization condition, is equivalent to the condition that $\sum_i e_i \otimes e_i$ is an eigenvector for such a matrix $A$ in the perturbation semigroup with eigenvalue $1$.
\end{remark}

\begin{lemma}
For $A = \sum C_{ij,kl} e_{ij}\otimes e_{kl}^{\circ}$ the self-adjointness condition is equivalent to demanding
$$C_{ij,kl} = \overline{C_{lk,ji}}.$$
\end{lemma}
\begin{proof}
The condition $A = A^*$ becomes
$$
 \sum C_{ij,kl} e_{ij} \otimes e_{kl}^{\circ} \\
 =  \sum \overline{C_{ij,kl}}e_{kl}^* \otimes e_{ij}^{* \circ} \\
 =  \sum \overline{C_{ij,kl}} e_{lk} \otimes e_{ji}^{\circ}.
$$
If we now relabel the last expression, we get $\sum \overline{C_{lk,ji}} e_{ij} \otimes e_{kl}^{\circ},$ so that $C_{ij,kl} = \overline{C_{lk,ji}}$.
\end{proof}
We now have the following proposition.
\begin{proposition}
\label{prop:comm.omega.genmat}
Let $A = \sum C_{ij,kl} e_{ij} \otimes e_{kl}^{\circ}$. Then $C_{ij,kl} = \overline{C_{lk,ji}}$ if and only if $\widehat{\Omega}\overline{A} = A \widehat{\Omega}$ with $\widehat{\Omega} = \sum e_{ij} \otimes e_{ij}^{\circ}\in M_N(\mathbb{C}) \otimes M_N(\mathbb{C})^{\circ}$.
\end{proposition}
\begin{proof}
We can write $\widehat{\Omega}$ as $\widehat{\Omega} = \sum \delta_m^r\delta_n^s e_{mn}\otimes e_{rs}^{\circ}$. Starting with the right hand side of the equation we get
$$\begin{array}{rcl}
A \widehat{\Omega} & = & (\sum C_{ij,kl} e_{ij}\otimes e_{kl}^{\circ}) (\sum \delta_m^r\delta_n^s e_{mn}\otimes e_{rs}^{\circ} )\\
& = & \sum C_{ij,kl} \delta_m^r\delta_n^s e_{ij}e_{mn} \otimes (e_{rs}e_{kl})^{\circ} \\
& = & \sum C_{ij,kl} \delta_m^r\delta_n^s\delta_j^m\delta_s^k e_{in} \otimes e_{rl}^{\circ} \\

& = & \sum C_{ij,kl} e_{ik} \otimes e_{jl}^{\circ}.
\end{array}$$
The left hand side of the equation reads
$$\begin{array}{rcl}
\widehat{\Omega} \overline{A} & = & (\sum \delta_m^r \delta_n^s e_{mn} \otimes e_{rs}^{\circ})(\sum \overline{C_{ij,kl}} e_{ij} \otimes e_{kl}^{\circ}) \\
& = & \sum \overline{C_{ij,kl}} \delta_m^r \delta_n^s e_{mn}e_{ij} \otimes (e_{kl}e_{rs})^{\circ} \\
& = & \sum \overline{C_{ij,kl}} \delta_m^r \delta_n^s \delta_i^n e_{mj} \otimes \delta_r^l e_{ks}^{\circ} \\

& = & \sum \overline{C_{ij,kl}} e_{lj} \otimes e_{ki}^{\circ} \\
& = & \sum \overline{C_{lk,ji}} e_{ik} \otimes e_{jl}^{\circ}.
\end{array}$$
Thus we have $C_{ij,kl} = \overline{C_{lk,ji}}$ if and only if $\widehat{\Omega}\overline{A} = A \widehat{\Omega}$.
\end{proof}
We now make the following identification
\begin{align}
\label{ident:genmat}
M_N(\mathbb{C}) \otimes M_N(\mathbb{C})^{\circ} & \rightarrow M_{N^2}(\mathbb{C}), \qquad  e_{ij} \otimes e_{kl}^{\circ} \mapsto e_{ij} \otimes e_{lk}, 
\end{align}
after which we can bring $\widehat{\Omega}$ into a more appealing form as a block matrix:
$$\widehat{\Omega} = \sum e_{ij} \otimes e_{ji}.$$
\begin{lemma}
The eigenvectors of $\widehat{\Omega}$ are given by $e_k \otimes e_l \pm e_l \otimes e_k$ with respective eigenvalues $1$ (for any $k,l=1,\ldots, N$) and $-1$ (for $k \neq l$).
\end{lemma}
\begin{proof}
This follows by elementary matrix multiplication:
$$
\sum e_{ij}e_k \otimes e_{ji}e_l \pm \sum e_{ij}e_l \otimes e_{ji}e_k
=  e_l \otimes e_k \pm e_k\otimes e_l. \qedhere
$$

\end{proof}
Since $e_k \otimes e_k$ is an eigenvector with eigenvalue $1$ for all $k$, we see that their sum must be an eigenvector with eigenvalue $1$ as well, i.e. we have
$$\widehat{\Omega} (\sum_i e_i \otimes e_i) = \sum_i e_i \otimes e_i.$$
We change to a basis consisting of eigenvectors, where we take $\sum e_i \otimes e_i$ to be identified with $e_1$ in the new basis. This gives us
\begin{equation}
\label{eqref:Omega}
\Omega = \begin{pmatrix}
I_{N(N+1)/2} & 0 \\
0 & -I_{N(N-1)/2} \\
\end{pmatrix}.
\end{equation}
As we have seen before $\sum e_i \otimes e_i$ in terms of the old basis is an invariant vector of a matrix $A$ in the perturbation semigroup. Thus, in the new basis the vector $e_1$ is an invariant vector for such a matrix $A$. We summarize the above results by the following
\begin{proposition}
We have
\begin{equation}
\label{eq:pertmatrix}
\Pert(M_N(\mathbb{C})) \cong \Big \lbrace A \in M_{N^2}(\mathbb{C}) \mid A e_1 = e_1,\; \Omega \overline{A} = A \Omega \Big \rbrace
\end{equation}
where
$$
e_1 = \begin{pmatrix} 1 \\ 0 \\ \vdots \\ 0 \end{pmatrix}, \qquad 
\Omega = \begin{pmatrix}
I_{N(N+1)/2} & 0 \\
0 & -I_{N(N-1)/2} \\
\end{pmatrix}.$$
The semigroup structure is given by matrix multiplication.
\end{proposition}

This allows for the following explicit description of elements in the perturbation semigroup $\Pert(M_N(\C))$. 
Let $A \in M_{N^2}(\mathbb{C})$ with $Ae_1 = e_1$, then we get that
\begin{equation}
\label{eq:block-pert-MN}
A = \begin{pmatrix}
1 & v \\
0 & B \\
\end{pmatrix},
\end{equation}
where $v$ is a row vector of length $N^2 -1$, while $B \in M_{N^2-1}(\mathbb{C})$. The condition that $\Omega \overline{A} = A \Omega$ then implies that
$$\Omega' \overline{B} = B \Omega',$$
and
$$\overline{v} = v \Omega'.$$
in terms of the matrix 
$$\Omega' = \begin{pmatrix}
I_{N(N+1)/2 -1} & 0 \\
0 & -I_{N(N-1)/2} \\
\end{pmatrix} = \begin{pmatrix}
I_{(N+2)(N-1)/2} & 0 \\
0 & -I_{N(N-1)/2} \\
\end{pmatrix}.$$
If we work this out we see that
$$v = \begin{pmatrix}
v_1 & iv_2 \\
\end{pmatrix},$$
where $v_1$ and $v_2$ are real row vectors of length $(N+2)(N-1)/2$ and $N(N-1)/2$, respectively. We also see that 
$$B = \begin{pmatrix}
B_1 & iB_2 \\
iB_3 & B_4 \\
\end{pmatrix},$$
where $B_1 \in M_{(N+2)(N-1)/2}(\mathbb{R}), \ldots, B_4 \in M_{N(N-1)/2}(\mathbb{R})$. 

This motivates the definition of a real vector space $V$ and semigroup $S$ by
\begin{align*}
V & =  \big \lbrace v \in \mathbb{C}^{N^2-1} \mid \overline{v} = v \Omega '  \big \rbrace, \\
S  &=  \big \lbrace B \in M_{N^2-1}(\mathbb{C}) \mid \Omega ' \overline{B} = B \Omega ' \big \rbrace.
\end{align*}
and to consider the semidirect product $V \rtimes S$ of $V$ and $S$. The semigroup law of $V \rtimes S$ is given by
$$(v,B)\cdot(v',B') = (v' + vB', BB').$$
\begin{proposition}
\label{prop:semidirect-prod-pert-MN}
For $V$ and $S$ as above we have an isomorphism of semigroups:
$$\Pert(M_N(\mathbb{C})) \cong V \rtimes S.$$
\end{proposition}
\begin{proof}
Let $A,A' \in \Pert(M_N(\mathbb{C}))$, then we have
\begin{equation*}
A = \begin{pmatrix}
1 & v \\
0 & B \\
\end{pmatrix}, \; A' = \begin{pmatrix}
1 & v' \\
0 & B' \\
\end{pmatrix}
\end{equation*}
for suitable $v,v' \in V$ and $B,B' \in S$. If we now multiply $A$ and $A'$ we get
$$AA' = \begin{pmatrix}
1 & v \\
0 & B \\
\end{pmatrix}\begin{pmatrix}
1 & v' \\
0 & B' \\
\end{pmatrix} = \begin{pmatrix}
1 & v' + vB \\
0 & BB' \\
\end{pmatrix}.$$
This coincides with the semigroup law in $V \rtimes S$, thus completing the proof. 
\end{proof}


\begin{corl}
Let $V$ be as above and define the group $G$ as 
$$G = \big \lbrace A \in GL_{N^2-1}(\mathbb{C}) \mid \Omega ' \overline{A} = A \Omega ' \big \rbrace.$$
Then the invertible elements in $\Pert(M_N(\C))$ form the semidirect product group 
$$\Pert(M_N(\mathbb{C}))^{\times} \cong V \rtimes G.$$
\end{corl}
\proof
This follows at once from Proposition \ref{prop:semidirect-prod-pert-MN} and the fact that 
$$(V \rtimes S)^{\times} = V \rtimes S^{\times}$$
which holds for any semigroup $S$ acting linearly on a vector space $V$. 
\endproof

As in the previous section, we show how the unitary group $U(N) = \U(M_N(\C))$ maps to this group of invertible elements. Again, there is a group homomorphism
\begin{equation}
\label{eq:group-homo-UN}
\begin{aligned}
U(N) &\to \Pert(M_N(\C))^\times\\
u& \mapsto u \otimes u^{*\circ}.
\end{aligned}
\end{equation}
The corresponding element $u \otimes \overline u \in M_N(\C) \otimes M_N(\C)$ defines the representation of $u \in U(N)$ on the tensor product $\C^N \otimes \overline{\C^N}$. Moreover, the block form that we found in \eqref{eq:block-pert-MN} parallels the decomposition of this representation into irreducible representations of $U(N)$:
$$
\C^N \otimes \overline{\C^N} \cong \C \oplus \C^{N^2-1}.
$$
Here $\C$ is the trivial representation space and $\C^{N^2-1}$ is the complexified adjoint representation space $\mathfrak{su}(N)_\C$ of $U(N)$.

\subsection{Perturbation semigroup of real matrix algebras}
Now that we have the semigroup $\Pert(M_N(\mathbb{C}))$ we consider the perturbation semigroup of the real matrix algebras $M_N(\mathbb{R})$ and $M_N(\mathbb{H})$. 
\subsubsection{$\mathcal{A}=M_N(\mathbb{R})$}
In order to determine the perturbation semigroup for $M_N(\mathbb{R})$ we can search the results we found for $\Pert(M_N(\mathbb{C}))$ for complex conjugation and subsequently ignore it. This means we get
$$\Pert(M_N(\mathbb{R})) \cong \Big \lbrace A \in M_{N^2} (\mathbb{R}) \mid A e_1 = e_1,\; \Omega A = A \Omega \Big \rbrace,$$
where
$$\Omega = \begin{pmatrix} 
I_{N(N+1)/2} & 0 \\
0 & -I_{N(N-1)/2} \\
\end{pmatrix}.$$
\begin{proposition}
\label{theo:verder.result.M_N(R)}
We have an isomorphism of semigroups:
$$\Pert(M_N(\mathbb{R})) \cong \Big(\mathbb{R}^{(N-1)(N+2)/2} \rtimes M_{(N-1)(N+2)/2}(\mathbb{R})\Big) \times M_{N(N-1)/2}(\mathbb{R}).$$
\end{proposition}
\proof
The conditions for a matrix $A \in M_{N^2}(\R)$ to be in $\Pert(M_N(\R)$ brings it in the following general form:
\begin{equation}
\label{eq:A-gen-form-MNR}
A= \begin{pmatrix}
1 & v_1 & 0 \\
0 & B_1 & 0 \\
0 & 0 & B_2 \\
\end{pmatrix}
\end{equation}
from which the proof of the statement follows. 
\endproof

\begin{corl}
The invertible elements of $\Pert(M_N(\mathbb{R}))$ are given by
$$\Pert(M_N(\mathbb{R}))^{\times} \cong \left( \mathbb{R}^{(N-1)(N+2)/2} \rtimes GL_{(N-1)(N+2)/2}(\mathbb{R}) \right) \times GL_{N(N-1)/2}(\mathbb{R})$$
\end{corl}
\proof
From the general form of $A$ in Equation \eqref{eq:A-gen-form-MNR} it follows that $\det(A) = \det(B_1)\det(B_2)$. Hence $A$ is invertible if and only if both $B_1$ and $B_2$ are invertible. 
\endproof

Again, let us consider the map from the unitary group $\U(M_N(\R))$ to $\Pert(M_N(\R))$. Actually, $\U(M_N(\R)) \cong O(N)$ and we have a map
\begin{align*}
O(N) &\to \Pert(M_N(\R))\\
u &\to u \otimes (u^{\sf T})^\circ.
\end{align*}
Upon identifying $M_N(\R)^\circ$ with $M_N(\R)$, the element $u \otimes (u^{\sf T})^\circ$ corresponds to the element $u \otimes u \in M_N(\R) \otimes M_N(\R)$. Hence, this defines a representation on the tensor product $\C^N \otimes \C^N$ of two copies of (the complexification of) the defining representation of $O(N)$. As opposed to the unitary groups encountered earlier, this tensor product has the following decomposition as $O(N)$-representations:
$$
\C^N \otimes \C^N \cong \C \oplus \C^{(N+2)(N-1)/2} \oplus \C^{N(N-1)/2}.
$$
The first summand is the trivial representation space of $O(N)$ (spanned by the vector $\sum_i e_i \otimes e_i$), the second consists of the symmetric tensors (spanned by the vectors of the form $e_i \otimes e_j + e_j \otimes e_i$) and the third consists of the skew-symmetric tensors (spanned by the vectors $e_i \otimes e_j - e_j \otimes e_i$). This gives rise to the dimensions $1, (N+2)(N-1)/2$ and $N(N-1)/2$ in the above decomposition. Moreover, this decomposition agrees with the block matrix form of $A$ in Equation \eqref{eq:A-gen-form-MNR}.

\subsubsection{$\mathcal{A}=\mathbb{H}$}
\label{sect:H}
We determine the perturbation semigroup of the quaternions $\mathbb{H}$. A convenient characterization of $\H$ is as the following set of $2 \times 2$ matrices:
$$\mathbb{H}= \Bigg\lbrace \begin{pmatrix}
\alpha & \beta\\
-\overline{\beta} & \overline{\alpha}\\
\end{pmatrix}
\mid \alpha,\beta \in \mathbb{C} \Bigg\rbrace.$$
Equivalently, for $A \in M_2(\mathbb{C})$ to be in $\mathbb{H}$ it should satisfy the condition $\widehat{J}\overline{A} = A \widehat{J}$ where
$$\widehat{J} = \begin{pmatrix}
0 & 1 \\
-1 & 0 \\
\end{pmatrix} = e_{12} - e_{21}.$$
The quaternions thus form a real subalgebra of $M_2(\mathbb{C})$ and in order to determine $\Pert(\H)$ we can start by looking at the matrices that form $\Pert(M_2(\mathbb{C}))$. Recall that the general form of elements therein was given by
$$A = \begin{pmatrix}
x_1				& z_2	& \overline{z_2}	& 1-x_1				\\
z_1				& z_4	& \overline{z_5}	& -z_1				\\
\overline{z_1}	& z_5	& \overline{z_4}	& -\overline{z_1}	\\
x_2				& z_3	& \overline{z_3}	& 1-x_2				\\
\end{pmatrix},$$
where $z_i \in \mathbb{C}\text{, for } i=1,\hdots,5, \; x_1,x_2 \in \mathbb{R}$.
We impose the commutation relation with $\widehat{J}$ in order to get a matrix in $\Pert(\mathbb{H})$. In fact, $\widehat{J}$ extends to the tensor product $M_2(\mathbb{C})\otimes M_2(\mathbb{C})^{\circ}$: for $A \otimes B^{\circ} \in M_2(\mathbb{C})\otimes M_2(\mathbb{C})^{\circ}$ to be in $\mathbb{H}\otimes \mathbb{H}^{\circ}$ we need
$$(\widehat{J}\otimes \widehat{J}^{\circ})\overline{(A \otimes B^{\circ})} = (\widehat{J}\otimes \widehat{J}^{\circ})(\overline{A} \otimes \overline{B}^{\circ}) = (A \otimes B^{\circ})(\widehat{J}\otimes \widehat{J}^{\circ}).$$
Once again using the identification 
\begin{align*}
M_2(\mathbb{C}) \otimes M_2(\mathbb{C})^{\circ} & \rightarrow M_4(\mathbb{C}), \qquad e_{ij} \otimes e_{kl}^{\circ} \mapsto e_{ij} \otimes e_{lk},
\end{align*}
we find that
\begin{equation}
\label{eq:quatJ}
\widehat{J} \otimes \widehat{J}^{\circ} \mapsto \widetilde{J} = (e_{12} - e_{21}) \otimes (e_{12} - e_{21})^{\sf T} = \begin{pmatrix}
0 & 0 & 0 & -1 \\
0 & 0 & 1 & 0 \\
0 & 1 & 0 & 0 \\
-1 & 0 & 0 & 0 \\
\end{pmatrix}.
\end{equation}
So for a matrix $A\in \Pert(\mathbb{H})$ we need to have $\widetilde{J}\overline{A} = A\widetilde{J}.$
In other words, for the matrix $A \in \Pert(M_2(\C))$ to be in $\Pert(\H)$ it should be of the form
$$A = \begin{pmatrix}
x & z_2	& \overline{z_2} & 1-x \\
z_1 & z_3 & z_4 & -z_1 \\
\overline{z_1} & \overline{z_4} & \overline{z_3} & -\overline{z_1} \\
1-x & -z_2 & -\overline{z_2} & x  \\
\end{pmatrix},$$
where $x \in \mathbb{R}, z_1,z_2,z_3,z_4 \in \mathbb{C}$. Since this is the same form as for $A \in \Pert(M_2(\mathbb{C}))$ it follows that we have a similar commutation relation for this $A$ with $\widehat{\Omega}$, namely $\widehat{\Omega} \overline{A} = A \widehat{\Omega}$. 

As in the case of $\Pert(M_2(\C))$ we can diagonalize $\widehat{\Omega}$ to 
$$
\Omega = \begin{pmatrix}
I_3 & 0 \\
0 & -1 \\
\end{pmatrix}
$$
where the new basis consists of eigenvectors, given by
$$e_1 \otimes e_1 \pm e_2 \otimes e_2,$$
$$e_1 \otimes e_2 \pm e_2 \otimes e_1.$$
We also write $\widetilde{J} = \widehat{J} \otimes \widehat{J}^{\circ}$ in terms of this new basis. Since 
\begin{gather*}
\big( (e_{12} - e_{21}) \otimes (e_{21}-e_{12})\big) (e_1 \otimes e_1 \pm e_2 \otimes e_2)  \\
 \qquad= \mp (e_1 \otimes e_1 \pm e_2 \otimes e_2),\\
\big( (e_{12} - e_{21}) \otimes (e_{21}-e_{12})\big) (e_1 \otimes e_2 \pm e_2 \otimes e_1) 
\\ \qquad= \pm (e_1 \otimes e_2 \pm e_2 \otimes e_1),
\end{gather*}
we retrieve the following expression for $\widetilde{J}$ in terms of the new basis
$$J = \begin{pmatrix}
-1 & 0 & 0 & 0 \\
0 & 1 & 0 & 0 \\
0 & 0 & 1 & 0 \\
0 & 0 & 0 & -1 \\
\end{pmatrix}.$$
With this we can find the general expression for $\Pert(\mathbb{H})$.
\begin{proposition}
In the above notation,
$$\Pert(\mathbb{H}) \cong \big\lbrace A \in M_4(\mathbb{C}) \mid Ae_1 = e_1, \Omega \overline{A} = A \Omega, J \overline{A} = A J \big \rbrace.$$
\end{proposition}
\begin{proof}
The conditions $Ae_1 = e_1$ and $\Omega \overline{A} = A \Omega$ ensure that the matrix is in $\Pert(M_2(\mathbb{C}))$, while $J \overline{A} = A J $ ensures that such a matrix is in fact an element of $\mathbb{H}\otimes \mathbb{H}^{\circ}$. 
\end{proof}
 Since $\Omega$ and $J$ have the same commutation relation with $A \in \Pert(\mathbb{H})$,  also the sum and difference of $\Omega$ and $J$ must have this commutation relation with $A$. We define $\Upsilon = (\Omega - J)/2 = e_{11}$ and $\Gamma = (\Omega + J)/2 = e_{22} + e_{33} - e_{44}$, {\it i.e.}
\begin{align*}
\Upsilon &= (\Omega - J)/2 =
\begin{pmatrix}
1 & 0 & 0 & 0 \\
0 & 0 & 0 & 0 \\
0 & 0 & 0 & 0 \\
0 & 0 & 0 & 0 \\
\end{pmatrix},\\
\Gamma &= (\Omega + J)/2 = \begin{pmatrix}
0 & 0 & 0 & 0 \\
0 & 1 & 0 & 0 \\
0 & 0 & 1 & 0 \\
0 & 0 & 0 & -1 \\
\end{pmatrix}.
\end{align*}

 \begin{proposition}
We have 
 $$\Pert(\mathbb{H}) \cong \Big\lbrace A \in M_4(\mathbb{C}) \mid Ae_1 = e_1, \Upsilon \overline{A} = A \Upsilon, \Gamma \overline{A} = A \Gamma \Big \rbrace.$$
 \end{proposition}
This readily leads to the following explicit characterization of elements in $\Pert(\H)$:
\begin{equation}
\label{eq:block-matrix-H}
A = \begin{pmatrix}
1 & 0 & 0 & 0 \\
0 & x_1 & x_2 & iy_1 \\
0 & x_3 & x_4 & iy_2 \\
0 & iy_3 & iy_4 & x_5 \\
\end{pmatrix},
\end{equation}
where $x_1,\hdots,x_5,y_1,\hdots,y_4 \in \mathbb{R}$. 
The form of the $3 \times 3$ block matrix is dictated by the following matrix $\Gamma'$:
$$
\Gamma'= \begin{pmatrix}
1 & 0 & 0 \\
0 & 1 & 0 \\
0 & 0 & -1 \\
\end{pmatrix}.
$$
\begin{proposition}
The perturbation semigroup for $\mathbb{H}$ is given by
$$
\Pert(\mathbb{H}) \cong \lbrace A \in M_3(\mathbb{C}) \mid \Gamma'\overline{A} = A \Gamma' \rbrace.
$$
\end{proposition}

\begin{corollary}
The group of invertible elements in $\Pert(\H)$ is
$$\Pert(\mathbb{H})^{\times} \cong \lbrace A \in GL_3(\mathbb{C}) \mid \Gamma' \overline{A} = A \Gamma' \rbrace.$$
\end{corollary}

The unitary group $\U(\H)$ is $SU(2)$, so that there is a map
$$
SU(2) \to \Pert(\H)
$$
similar to the map $U(2) \to \Pert(M_2(\C))$ in Section \ref{subsec:M_2C}. Again, the block form of $A$ in Equation \eqref{eq:block-matrix-H} corresponds to the decomposition of the representation of $SU(2)$ on $\C^2 \otimes \overline{\C^2}$ into irreducible summands $\C$ and $\C^3 \cong \mathfrak{su}(2)_\C$.

\subsubsection{$\mathcal{A}=M_N(\mathbb{H})$}
Finally, we determine the perturbation semigroup for $M_N(\mathbb{H})$. A matrix in $\Pert(M_N(\mathbb{H}))$ is characterized by a matrix similar to $\widetilde{J}$ that we had for $\Pert(\mathbb{H})$. In fact, we have the following 
\begin{lemma}
The perturbation semigroup of $\Pert(M_N(\mathbb{H}))$ can be obtained from $\Pert(M_{2N}(\mathbb{C}))$ as follows:
$$
\Pert(M_N(\H)) \cong \left\{ A \in \Pert(M_{2N}(\C)):  \widetilde{L}\overline{A} = A \widetilde{L} \right\}
$$
with $\widetilde{L}= I_{N^2} \otimes \widetilde{J}$.
\end{lemma}
\begin{proof}
Let $A \in M_{4N^2}(\mathbb{C})$ and let $\widetilde{L} = I_{N^2} \otimes \widetilde{J}$. 
We write $A = \sum e_{ij} \otimes A_{ij}$ for $A_{ij} \in \H$ and $i,j=1,\ldots,N$. 
It is clear that $\widetilde{L} \overline{A} = A \widetilde{L}$ amounts to imposing $\widetilde{J} \overline{A_{ij}} = A_{ij} \widetilde{J}$ for all $i,j$. In other words, this amounts to $A_{ij}$ to be in $\mathbb{H} \otimes \mathbb{H}^{\circ}$ hence completing the proof.
\end{proof}

We now want to simultaneously diagonalize $\widehat{\Omega}$ and $\widetilde{L}$, just as we did for $\mathbb{H}$. Note that $\widehat{\Omega} = \widehat{\Omega}_{N} \otimes \widehat{\Omega}_{2}$ in terms of the matrices of Equations \eqref{eqref:Omega} and \eqref{eq:omega-2} for $M_N(\C)$ and $M_2(\C)$, respectively. Hence, diagonalizing $\widehat{\Omega}_2$ and $\widetilde{J}$ as in Section \ref{sect:H}, we can write:
\begin{align*}
\Omega &= \begin{pmatrix} I_{N(N+1)/2} & 0 \\ 0 & -I_{N(N-1)/2} \end{pmatrix} \otimes \begin{pmatrix} 1 & 0 & 0 & 0 \\ 0 & 1 & 0 & 0 \\ 0 & 0 & 1 & 0 \\ 0 & 0 & 0 & -1\end{pmatrix},\\
L &= \begin{pmatrix} I_{N(N+1)/2} & 0 \\ 0 & I_{N(N-1)/2} \end{pmatrix} \otimes   \begin{pmatrix} -1 & 0 & 0 & 0 \\ 0 & 1 & 0 & 0 \\ 0 & 0 & 1 & 0 \\ 0 & 0 & 0 & -1\end{pmatrix}.
\intertext{These matrices are equivalent to the following diagonal matrices in $M_{4N^2}(\C)$:}
\Omega &= \begin{pmatrix}
I_{N^2} & 0 & 0 & 0 \\
0 & -I_{N(N-1)} & 0 & 0 \\
0 & 0 & -I_{N^2} & 0 \\ 
0 & 0 & 0 & I_{N(N+1)}  
\end{pmatrix},\\
L &= \begin{pmatrix}
-I_{N^2} & 0 & 0 & 0 \\
0 &  I_{N(N-1)} & 0 & 0 \\
0 & 0 & -I_{N^2} & 0 \\ 
0 & 0 & 0 & I_{N(N+1)}  
\end{pmatrix}.
\end{align*}
We thus get $$\Pert(M_N(\mathbb{H})) \cong \Big\lbrace A \in M_{4N^2}(\mathbb{C}) \mid Ae_1 = e_1, \Omega \overline{A} = A \Omega, L \overline{A} = A L \Big \rbrace.$$
As up to conjugation $A$ commutes with both $\Omega$ and $L$, every linear combination of the latter two must satisfy a similar commutation relation with $A$. We introduce block-diagonal matrices
\begin{align*} 
\Psi &= (\Omega - L)/2 = :
\begin{pmatrix} 1 & 0 & 0 \\ 0 & \Psi' & 0 \\ 0 & 0 & 0_{N(2N+1)} \end{pmatrix},
\intertext{and} 
\Theta &= (\Omega + L)/2 =: 
\begin{pmatrix} 0_{N(2N-1)} & 0 \\ 0 & \Theta' \end{pmatrix},
\end{align*}
where we have implicitly defined matrices $\Psi' \in M_{(2N+1)(N-1)}(\C)$ and $\Theta' \in M_{N(2N+1)}(\C)$ by
\begin{align*}
\Psi' = \begin{pmatrix}
I_{N^2-1} & 0 \\
0 & -I_{N(N-1)} \\
\end{pmatrix},
\qquad
\Theta' = \begin{pmatrix}
I_{N^2} & 0 \\
0 & -I_{N(2N-1)} \\
\end{pmatrix}.
\end{align*}
The reason for this particular block decomposition will become clear in the following proposition; first note that we have by linearity
$$\Pert(M_N(\mathbb{H})) \cong \Big\lbrace A \in M_{4N^2}(\mathbb{C}) \mid Ae_1 = e_1, \Psi \overline{A} = A \Psi, \Theta \overline{A} = A \Theta \Big \rbrace.$$
\begin{proposition}
We have
$$
\Pert(M_N(\mathbb{H})) \cong \left( V \rtimes S \right) \times T ,
$$
where
\begin{align*}
V &= \Big\lbrace v \in \C^{(2N+1)(N-1)} \mid \overline{v} = v \Psi'  \Big\rbrace,\\
S &= \Big\lbrace B \in M_{(2N+1)(N-1)}(\mathbb{C}) \mid \Psi' \overline{B} = B \Psi' \Big\rbrace,\\
T &= \Big\lbrace C \in M_{N(2N+1)} (\mathbb{C}) \mid \Theta' \overline{C} = C \Theta' \Big\rbrace.
\end{align*}
\end{proposition}
\begin{proof}
Let us start with a matrix $A \in M_{4N^2}(\C)$ that leaves $e_1$ invariant and write $A$ in the following suggestive form
$$
A = \begin{pmatrix}
1 & v & w \\
0 & B &B' \\
0 & C' & C
\end{pmatrix},$$
where $B \in M_{(2N+1)(N-1)}(\mathbb{C})$ and $C \in M_{N(2N+1)}(\mathbb{C})$ and the other block matrices $B',C'$ and the vectors $v,w$ chosen in a compatible way. Applying the commutation relation of $A$ with $\Psi$ gives 

$$A = \begin{pmatrix}
1 & v & 0 \\
0 & B & 0 \\
0 & 0 & C \\
\end{pmatrix},
$$
with $\overline{v} = v \Psi' $, $\Psi' \overline{B} = B \Psi$ and $\Theta' \overline{C} = C \Theta'$. 
\end{proof}
Note that this result is in concordance with the perturbation semigroup that we have found previously for $\mathbb{H}$, with the semidirect product vanishing for $N=1$.

\medskip

Again we can trace the unitary elements $u$ in $M_N(\H)$ in $\Pert(M_N(\H))$ ({\it cf.} Proposition \ref{prop:UA-Pert}). Note that $\U(M_N(\H))$ can be identified with the Lie group $Sp(N) = Sp(2N,\C) \cap U(2N)$, which is of dimension $N(2N+1)$. Then, the above block-diagonal decomposition of $A \in \Pert(M_N(\H))$ corresponds to the decomposition of the tensor product representation $\C^{2N} \otimes \overline{\C^{2N}}$ of $Sp(N)$ into irreducible representations. We find, for example, the (complexification of the) adjoint representation on $\mathfrak{sp}(N)_\C$ via the lower-diagonal matrix $C \in GL_{N(2N+1)}(\C)$.

\end{document}